\documentclass[conference,letterpaper]{IEEEtran}
\usepackage[utf8]{inputenc} 
\usepackage[T1]{fontenc}
\usepackage{url}
\usepackage{amsfonts}                                                          
\usepackage{epsfig}
\usepackage{amsthm}
\usepackage{graphicx}        
\usepackage{latexsym}
\usepackage{amssymb}
\usepackage{amsmath}
\usepackage{multirow}
\usepackage{cite}
\usepackage{mathtools}
\usepackage{epstopdf}
\usepackage{etoolbox}
\usepackage{array}
\usepackage{mathptmx}
\usepackage[usenames,dvipsnames]{xcolor}
\usepackage{hyperref}
\usepackage[bottom]{footmisc}
\usepackage{tikz}

\usetikzlibrary{decorations.pathreplacing}
\newcommand\numeq[1]%
  {\stackrel{\scriptscriptstyle(\mkern-1.5mu#1\mkern-1.5mu)}{=}}
\newcommand\geab[1]%
  {\stackrel{\scriptscriptstyle(\mkern-1.5mu#1\mkern-1.5mu)}{\ge}}
\usepackage{verbatim}
\usepackage{calrsfs}
\usepackage{listings}
\usepackage[short]{optidef}
\usepackage{soul}
\hypersetup{
    colorlinks=true,
    linkcolor=OliveGreen,
    filecolor=magenta,      
    urlcolor=blue,
    citecolor=Orange,
}

\DeclareMathAlphabet{\pazocal}{OMS}{zplm}{m}{n}

\DeclareMathOperator{\trace}{tr}
\DeclareMathOperator{\diag}{diag}
\newcommand{\Cov}{\mathrm{cov}}
\let\bbordermatrix\bordermatrix
\patchcmd{\bbordermatrix}{8.75}{4.75}{}{}
\patchcmd{\bbordermatrix}{\left(}{\left[}{}{}
\patchcmd{\bbordermatrix}{\right)}{\right]}{}{}

\newcommand{\T}{^{\mbox{\tiny T}}}

\newcommand{\sr}{\stackrel}

\newcommand{\rar}{\rightarrow}

\newcommand{\tri}{\sr{\triangle}{=}}

\newcommand{\be}{\begin{equation}}
\newcommand{\ee}{\end{equation}}
\newcommand{\bea}{\begin{eqnarray}}
\newcommand{\eea}{\end{eqnarray}}
\newcommand{\bes}{\begin{eqnarray*}}
\newcommand{\ees}{\end{eqnarray*}}
\newcommand{\bce}{\begin{center}}
\newcommand{\ece}{\end{center}}

\newcommand{\eeae}{\end{IEEEeqnarray}}

\def\VR{\kern-\arraycolsep\strut\vrule &\kern-\arraycolsep}
\def\vr{\kern-\arraycolsep & \kern-\arraycolsep}
\newtheorem{Algorithm}{Algorithm}[section]

\newcommand{\ben}{\begin{enumerate}}
\newcommand{\een}{\end{enumerate}}

\newcommand{\hso}{\hspace{.1in}}
\newcommand{\hst}{\hspace{.2in}}

\newtheorem{theorem}{Theorem}[section]

\newtheorem{remark}{Remark}[section]

\newtheorem{lemma}{Lemma}[section]

\newenvironment{list5}{
  \begin{list}{$\bullet$}{%
      \setlength{\itemsep}{0.05cm}
      \setlength{\labelsep}{0.22cm}
      \setlength{\labelwidth}{0.3cm}
      \setlength{\parsep}{0.0in} 
      \setlength{\parskip}{0.0in}
      \setlength{\topsep}{0in} 
      \setlength{\partopsep}{0.0in}
      \setlength{\leftmargin}{0.22in}}}
      {\end{list}}

 \title{Joint Rate Distortion Function of a Tuple of Correlated Multivariate Gaussian Sources with Individual Fidelity Criteria}
  \author{
   \IEEEauthorblockN{
     Evagoras Stylianou\IEEEauthorrefmark{1}, Charalambos D. Charalambous\IEEEauthorrefmark{2}, and
      Themistoklis Charalambous\IEEEauthorrefmark{3} 
     \\}
   \IEEEauthorblockA{
      \IEEEauthorrefmark{1}Department of Electrical and Computer Engineering, Technical University of Munich,\\
       \IEEEauthorrefmark{2}Department of Electrical and Computer Engineering,   University of Cyprus \\
   \IEEEauthorrefmark{3}Department of Electrical Engineering and Automation, School of Electrical Engineering, Aalto University
 \\
      Emails: evagoras.stylianou@tum.de, 
     chadcha@ucy.ac.cy,
        themistoklis.charalambous@aalto.fi}}

\begin{document}
\maketitle

%
%
%
%
\begin{abstract}
In this paper we analyze  the joint rate distortion function (RDF), for a tuple of correlated sources taking 
values in abstract alphabet spaces (i.e., continuous) subject to two individual distortion criteria. First, we  derive  structural properties of the realizations of the reproduction Random Variables (RVs), which induce the corresponding  optimal test channel distributions of the joint RDF. Second, we consider a tuple of  correlated multivariate jointly Gaussian RVs, $X_1 : \Omega \rar {\mathbb R}^{p_1},  X_2 : \Omega \rar  {\mathbb R}^{p_2}$ with two square-error fidelity criteria, and we derive additional structural properties of the optimal realizations, and   use these to characterize  the RDF as a convex optimization problem with respect to the parameters of the realizations. We  show that the computation of the joint RDF can be performed by semidefinite programming. Further, we derive closed-form expressions of the joint RDF, such that Gray's  \cite{gray1973} lower bounds hold with equality, and verify their consistency with the semidefinite programming computations.  
\end{abstract}

%
%
%
%
\section{Literature Review, Problem Formulation, and Main Contributions}

\subsection{Literature Review} \label{sec:Liter}

Gray  \cite[Theorem~3.1, Corollary~3.1]{gray1973} derived lower bounds on the joint rate distortion functions (RDFs), of  a tuple of Random Variables (RVs) taking values in arbitrary, abstract spaces,   $X_1: \Omega \rar {\mathbb X}_1$, $X_2: \Omega \rar {\mathbb X}_2$,  with a weighted distortion, expressed in terms of  conditional RDFs, and marginal RDFs. 
Gray and Wyner in \cite{sourcecoding}, characterized the rate distortion region of a tuple of correlated RVs,  using   the joint, conditional and marginal RDFs.  Xiao and Luo \cite[Theorem 6]{xiao} derived  the closed-form expression of the joint RDF for  a tuple of scalar-valued correlated Gaussian RVs, with two square-error distortion criteria, while 
 Lapidoth and Tinguely  \cite{lapidoth} re-derived Xiao's  and Luo's joint RDF using   an alternative method. Xu, Liu and Chen \cite{xu2015lossy} and Viswanatha, Akyol and Rose \cite{viswanatha2014lossy}, generalized Wyner's common information \cite{wyner1975common} to its lossy counterpart, as the minimum common  message rate on the Gray and Wyner rate region with sum rate equal to the joint RDF with two individual distortion functions. The  analysis in \cite{xu2015lossy,viswanatha2014lossy}, includes the application of a tuple of scalar-valued,  jointly Gaussian RVs. More recent work on rates that lie on the Gray and Wyner rate region are found in \cite{charalambous2020characterization}.

\subsection{Problem Formulation}
\subsubsection{The Joint RDF with Individual Distortion Functions}
This paper is concerned with the  joint RDF of a tuple of RVs taking values in abstract spaces  (i.e., continuous-valued RVs),  $X_1: \Omega  \rar  {\mathbb X}_1,  {X}_{2}: \Omega  \rar  {\mathbb X}_2$ of reconstructing $X_i$ by $\widehat{X}_i: \Omega  \rar  \widehat{\mathbb X}_i$, for  $i=1,2$, , subject to two  distortion functions   $d_{X_i} :{\mathbb X}_i \times \widehat{{\mathbb X}}_i \rar [0,\infty),i=1,2$,  defined by   
\begin{align} 
{R}_{X_1,X_2}(\Delta_1,\Delta_2) = \inf_{\pazocal{M}(\Delta_{1},\Delta_{2})} I(X_1,X_2;\widehat{X}_1,\widehat{X}_2) \label{jRDF_g}
\end{align}
where  $I(X_1,X_2;\widehat{X}_1,\widehat{X}_2)$ is the mutual information of RVs $(X_1,X_2)$ and $(\widehat{X}_1,\widehat{X}_2)$, the set $\pazocal{M}(\Delta_{1},\Delta_{2})$ is specified by
\begin{align}
&\pazocal{M}(\Delta_{1},\Delta_{2}) = \Big \{\widehat{X}_1: \Omega \rightarrow  \widehat{{\mathbb X}}_1,\;\widehat{X}_2: \Omega \rightarrow  \widehat{{\mathbb X}}_2 \Big |\; \mathbf{P}_{X_1,X_2,\widehat{X}_1,\widehat{X}_2}\;\;\mbox{has} \nonumber \\ &\mbox{$(X_1,X_2)$-marginal}\; \mathbf{P}_{X_1,X_2}, \mathbf{E} \big\{ d_{X_i}(X_i,\widehat{X}_i) \big \} \le \Delta_{i}, \; i=1,2  \Big \} \label{jRDF_g1}
\end{align}
and the level of distortions are $\Delta_i \in [0,\infty), i=1,2$.  The joint RDF characterizes  the infimum of all achievable  rates of a sequence of  rate distortion codes, $(f_E, g_D)$, as depicted in Figure~\ref{fig:jointRDF},  of  reconstructing $(X_1^n, X_2^n) \tri  \{(X_{1,t}, X_{2,t}): t=1,2, \ldots,n\}$, by  $(\widehat{X}_1^n, \widehat{X}_2^n) \tri  \{(\widehat{X}_{1,t}, \widehat{X}_{2,t}): t=1,2, \ldots,n\}$, where $ \widehat{X}_{i,t} : \Omega \rightarrow  \widehat{\mathbb X}_i  \;i=1,2,\;t=1,2, \ldots, n$ and ${\bf P}_{X_{1,t}, X_{2,t}}={\bf P}_{X_1,X_2}, \; \forall t$, with  distortion   $\frac{1}{n}{\bf E}\{d_{X_i} (X_i^n, \widehat{X}_i^n)\}\leq \Delta_i, i=1,2$, for sufficiently large $n$. The computation of  ${R}_{X_1,X_2}(\Delta_1,\Delta_2)$ is  indispensable in the characterization of the Gray and Wyner rate region, and in the above mentioned applications.  
\begin{figure}[t]
  \centering
  \includegraphics[width=0.98\columnwidth]{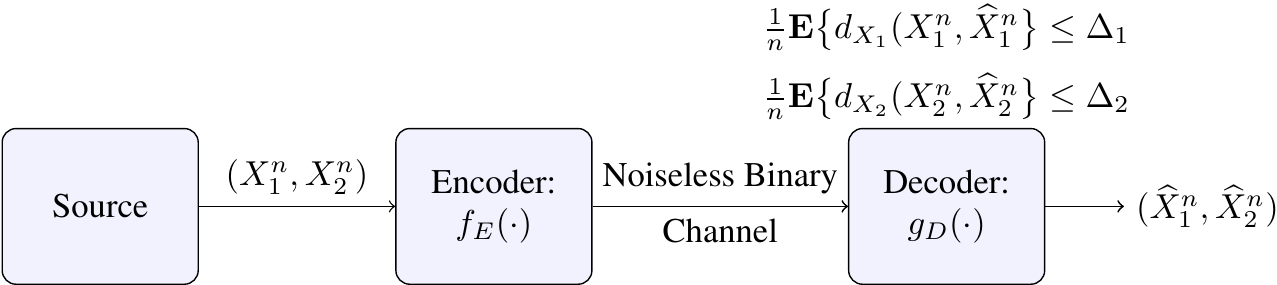}
  \vspace{-0.23cm}
\caption{Lossy Compression of correlated sources with individual distortion criteria\vspace{-0.4cm}.}
  \label{fig:jointRDF}
\end{figure}

Our first objective is to   identify {\it  structural properties} of  realizations of  the tuple of   RVs $(\widehat{X}_1,\widehat{X}_2)$ in the set   $\pazocal{M}(\Delta_1,\Delta_2)$, and structural properties of  corresponding   induced forward test channel distributions $\mathbf{P}_{\widehat{X}_1,\widehat{X}_2|X_1,X_2}$ or backward test channel distributions $\mathbf{P}_{X_1,X_2|\widehat{X}_1,\widehat{X}_2}$,  such that $\mathbf{E} \big\{ d_{X_i}(X_i,\widehat{X}_i) \big \} \le \Delta_i, \;i=1,2 $, i.e., to   characterize ${R}_{X_1,X_2}(\Delta_1,\Delta_2)$.

\renewcommand\arraystretch{1.1}
\subsubsection{The Joint RDF of a Tuple of Multivariate Gaussian Sources} 
Our second objective is to compute the joint RDF ${R}_{X_1,X_2}(\Delta_1,\Delta_2)$, of a tuple of jointly independent and identically distributed multivariate Gaussian RVs, $(X_1^n, X_2^n) \tri  \{(X_{1,t}, X_{2,t}): t=1,2, \ldots,n\}$,  where $ X_{i,t} : \Omega \rightarrow  {\mathbb R}^{p_i},\;i=1,2,\;t=1,2, \ldots, n$, i.e.,     ${\bf P}_{X_{1,t}, X_{2,t}}={\bf P}_{X_1,X_2}, \; \forall t$  is a multivariate jointly Gaussian distribution and denoted by  $(X_1, X_2)\in G(0,Q_{(X_1,X_2)})$, subject to two  square-error distortion  functions,  all defined by 
\begin{align}
&Q_{(X_{1,t}, X_{2,t})} =  {\mathbf E} \bigg\{ \begin{pmatrix} X_{1,t} \\ X_{2,t}  \end{pmatrix}  \begin{pmatrix}X_{1,t} \\ X_{2,t}  \end{pmatrix}\T \bigg\}=   \begin{pmatrix} Q_{X_1} & Q_{X_1,X_2} \\ Q_{X_1,X_2}\T & Q_{X_2} \end{pmatrix}  \label{prob_1}\\
& X_{1,t} \in G(0, Q_{X_1}), \hst  X_{2,t} \in G(0, Q_{X_2}), \hso \forall t,  \label{prob_2}\\
& \widehat{X}_{1,t}: \Omega  \rar  \widehat{\mathbb X}_1 \tri    {\mathbb R}^{p_1}, \hso \widehat{X}_{2,t}: \Omega  \rar  \widehat{\mathbb X}_2 \tri  {\mathbb R}^{p_2} \hso
\forall t,\label{prob_8} \\
& d_{X_i} (x_i^n, \widehat{x}_i^n)= \frac{1}{n} \sum_{t=1}^n ||x_{i,t}-\widehat{x}_{i,t}||_{{\mathbb R}^{p_{i}}}^2,\;\;\;i=1,2,  \label{prob_9}
\end{align}
where $p_i$ are  positive integers for $i=1,2$. Here $X \in G(0,Q_X)$ means $X$ is a Gaussian RV,  with zero mean and symmetric nonnegative definite covariance matrix $Q_X\succeq 0$.

\subsection{Main Contributions}

\begin{list5} 
\item[1)] The derivation of structural properties of   test channel distributions  $\mathbf{P}_{\widehat{X}_1,\widehat{X}_2|X_1,X_2}$, and  corresponding  realizations of the reproduction  RVs $(\widehat{X}_1,\widehat{X}_2)$ which induce these distributions, and characterize  $R_{X_1,X_2}(\Delta_1,\Delta_2)$. 

\item[2)] The characterization of  $R_{X_1,X_2}(\Delta_1,\Delta_2)$ for     jointly Gaussian multivariate sources, $X_1 : \Omega \rar {\mathbb R}^{p_1},  X_2 : \Omega \rar  {\mathbb R}^{p_2}$, with  square-error distortion criteria,  (\ref{prob_1})-(\ref{prob_9}),  parametrization of reproduction  RVs $(\widehat{X}_1,\widehat{X}_2)$ and corresponding test channels, and calculation of   $R_{X_1,X_2}(\Delta_1,\Delta_2)$ using convex numerical algorithms. Further, derivation of  closed-form expressions for $R_{X_1,X_2}(\Delta_1,\Delta_2)$, to verify the numerical algorithms.  This includes  the  distortion region ${\cal D}_{(X_1,X_2)}$, such that Gray's lower bound \cite{gray1973} holds with equality, 
 \begin{align}
 R_{X_1,X_2}(\Delta_1,\Delta_2)&=R_{X_1}(\Delta_1)+R_{X_2}(\Delta_2)-I(X_1;X_2). \label{equality}
 \end{align} 
\end{list5}
 the value of the RDF derived by 
Xiao and Luo \cite{xiao}. 
The tools used in this paper have been used to derive structural properties of the nonanticipative RDF of multivariate Gaussian Markov~\cite{ISIT:2020} and autoregressive~\cite{CDC:2019} processes.

\section{Properties of Realizations of Test Channels}
Let  ${\mathbb Z}$ and ${\mathbb Z}_+$ be the set of integers and positive integers, respectively. 
Let ${\mathbb R}$ be the set of real numbers. 
The expression $\mathbb{R}^{n \times m}$
denotes the set of $n$ by $m$ matrices with elements
 the real numbers, for $n, ~ m \in {\mathbb Z}_+$.
For the symmetric matrix $Q \in {\mathbb R}^{n \times n}$, 
inequality $Q \succ  0$ (resp. $Q \succeq 0$) means  the matrix is positive definite (resp. semi-definite).
The notation $Q_2 \succeq Q_1$ means that
$Q_2 - Q_1 \succeq  0$.
For any  matrix $A\in \mathbb{R}^{p\times m}, (p,m)\in {\mathbb Z}_+\times {\mathbb Z}_+$, we denote its transpose by $A\T$, and for $m=p$,  we denote its trace and its determinant  by  $\trace(A)$ and $\det\big(A\big)$, respectively. 
The $n$ by $n$ identity (resp. zero) matrix is represented by $I_n$ (resp. $0_n$). For matrix $A\in \mathbb{R}^{p\times p} $,  $\mathrm{diag}(A)$ is the  matrix with  diagonal entries those of  $A$ and zero elsewhere. $ \mathrm{Block}\text{-}\mathrm{diag}(A,B)$ is a square diagonal matrix in which the diagonal elements are square matrices $A\in \mathbb{R}^{p_1\times p_1} $ and $B\in \mathbb{R}^{p_2\times p_2} $, and the off-diagonal elements are zero.
Given a triple of real-valued RVs $X_i: \Omega \rar {\mathbb X}_i, i=1,2,3$, we say that RVs $(X_2, X_3)$ are conditional independent given RV $X_1$ if   ${\bf P}_{X_2, X_3|X_1}={\bf P}_{X_2|X_1}{\bf P}_{X_3|X_1} -$a.s (almost surely); the specification a.s is often omitted. 
 The mutual information between  RV $X$ and RV $Y$ is denoted by $I(X;Y)$.\\
The conditional covariance of the two-component vector RV $X = (X_1\T, X_2\T)\T$, $X_i: \Omega \rar {\mathbb R}^{p_i}, i=1,2$, conditioned on the two-component  vector $\widehat{X} = (\widehat{X}_1\T, \widehat{X}_2\T )\T$, $\widehat{X}_i: \Omega \rar {\mathbb R}^{p_i}, i=1,2$, is denoted by $Q_{(X_1,X_2)|\widehat{X}}\tri \Cov\Big(X,X\Big|\widehat{X}\Big) \succeq 0$, where
\vspace{-0.2cm}
\begin{align}
Q_{(X_1,X_2)|\widehat{X}} =& \begin{pmatrix}
Q_{X_1|\widehat{X}} & Q_{X_1,X_2|\widehat{X}}  \\
Q_{X_1,X_2|\widehat{X}}\T & Q_{X_2|\widehat{X}} \\
\end{pmatrix} \in {\mathbb R}^{(p_1+p_2)\times (p_1+p_2)}, \nonumber\\
Q_{X_1,X_2|\widehat{X}}  \tri &\Cov\Big(X_1,X_2\Big| \widehat{X}\Big). \nonumber \\
 \numeq 1 &  \mathbf{E} \Big\{ \Big( X_1 - {\mathbf E}\Big\{ X_1\Big|\widehat{X}\Big\} \Big) \Big(X_2 - {\mathbf E}\Big \{ X_2\Big|\widehat{X}\Big\} \Big)\T \Big\}  \nonumber \\
= &  \mathbf{E} \Big\{ E_1  E_2\T \Big\}, \hso E_i \tri X_i - {\mathbf E}\Big\{ X_i\Big|\widehat{X}\Big\}, \hso i=1,2 \nonumber
\end{align}
and where (1) holds if $(X_1, X_2, \widehat{X}_1, \widehat{X}_2)$ is jointly Gaussian. Similarly for $Q_{X_i|\widehat{X}}, i=1,2$. Consequently, for jointly Gaussian RVs $(X_1, X_2, \widehat{X}_1, \widehat{X}_2)$, and the  two-component vector  RV $E  \tri (E_1\T,E_2\T)\T$,  we have $Q_{(X_1,X_2)|\widehat{X}}=\Sigma_{(E_1,E_2)}$ (unconditional).

In Theorem~\ref{thm:jointlower} we identify a structural property of the tuple  $(\widehat{X}_1,\widehat{X}_2)$ to achieve a lower bound on $I(X_1,X_2; \widehat{X}_1, \widehat{X}_2)$, for any tuple of  RVs $(X_1,X_2)$ with arbitrary distribution ${\bf P}_{X_1,X_2}$. 
\begin{theorem}
 \label{thm:jointlower}
Let $(X_1, X_2, \widehat{X}_1,\widehat{X}_2)$ be arbitrary RVs taking values in  the abstract spaces ${{\mathbb X}}_1 \times {{\mathbb X}}_2\times \widehat{{\mathbb X}}_1\times \widehat{{\mathbb X}}_2$,  with arbitrary joint  distribution ${\bf P}_{X_1,X_2, \widehat{X}_1,\widehat{X}_2}$,  and ${{\mathbb X}}_1 \times {{\mathbb X}}_2-$joint marginal the fixed distribution ${\bf P}_{X_1,X_2}$ of $(X_1, X_2)$.\\ 
(a) Define 
\vspace{-0.2cm} 
\begin{align}
&\overline{X}_i^{\mathrm{cm}} = g_i^{\mathrm{cm}}\big (\widehat{X}_1,\widehat{X}_2\big) \tri   \mathbf{E} \Big\{ X_i\Big|\widehat{X} \Big\}, \quad i=1,2, \\
&g_i^{\mathrm{cm}}: \widehat{\mathbb X}_1\times \widehat{\mathbb X}_2 \rightarrow \widehat{\mathbb X}_i,\hso \mbox{$g_i^{\mathrm{cm}}(\cdot)$ are measurable functions,} \; i=1,2 \nonumber.
\end{align}
Then, the following inequality holds: 
\begin{align}
I(X_1,X_2;\widehat{X}_1,\widehat{X}_2) &\ge I\big (X_1,X_2;{{g}}_1^{\mathrm{cm}} (\widehat{X}_1,\widehat{X}_2),{{g}}_2^{\mathrm{cm}} (\widehat{X}_1,\widehat{X}_2)\big ).\label{eq:MutualInfoIneq} 
\end{align}
Moreover,  if there exist RVs $(\widehat{X}_1, \widehat{X}_2)$ such that  the functions $g_i^{\mathrm{cm}}(\cdot,\cdot)$ satisfy $g_i^{\mathrm{cm}} (\widehat{X}_1,\widehat{X}_2)  = \widehat{X}_i-\text{a.s}$ for $i=1,2$, then the inequality in \eqref{eq:MutualInfoIneq} holds with equality.\\
(b) Let  ${\mathbb X}_1 \times {\mathbb X}_2\times \widehat{{\mathbb X}}_1\times \widehat{{\mathbb X}}_2={\mathbb R}^{p_1} \times {\mathbb R}^{p_2}\times {\mathbb R}^{p_1}\times {\mathbb R}^{p_2}$, $p_1,p_2 \in {\mathbb Z}_+$. 
 For all measurable functions $g_i(\widehat{X}_1, \widehat{X}_2)$, $i=1,2$ then 
\begin{align}
&{\bf E}\Big\{\big|\big|X_i-g_i(\widehat{X}_1, \widehat{X}_2)\big|\big|_{{\mathbb R}^{p_i}}^2\Big\}\geq {\bf E}\Big\{\big|\big|X_i-\mathbf{E} \Big\{ X_i\Big|\widehat{X} \Big\}\big|\big|_{{\mathbb R}^{p_i}}^2\Big\}, i=1,2.\nonumber
\end{align}
(c) If  ${\mathbb X}_1 \times {\mathbb X}_2\times \widehat{{\mathbb X}}_1\times \widehat{{\mathbb X}}_2={\mathbb R}^{p_1} \times {\mathbb R}^{p_2}\times {\mathbb R}^{p_1}\times {\mathbb R}^{p_2}$, $p_1,p_2 \in {\mathbb Z}_+$,  $d_{X_i}(x_i,\widehat{x}_i) = || x_i-\widehat{x}_i||_{\mathbb{R}^{p_i}}^2,
\;i=1,2$,  $g_i^{\mathrm{cm}} (\widehat{X}_1,\widehat{X}_2)  = \widehat{X}_i-\text{a.s}, \;i=1,2$,    then the joint RDF  of (\ref{jRDF_g}) is characterized by 
\begin{align} 
{R}_{X_1,X_2}(\Delta_1,\Delta_2) = \inf_{\pazocal{M}^{\mathrm{cm}}(\Delta_{1},\Delta_{2})} I(X_1,X_2;\widehat{X}_1,\widehat{X}_2) \label{jRDF_g_cm}
\end{align}
where $\pazocal{M}^{\mathrm{cm}}(\Delta_{1},\Delta_{2})$ is specified by the subset  of  $ \pazocal{M}(\Delta_1,\Delta_2)$, with the additional restriction   $\widehat{X}_i= \mathbf{E} \Big\{ X_i\Big|\widehat{X} \Big\}, i=1,2$.
\end{theorem}

\begin{proof} (a) By properties of mutual information, we have 
\begin{align}
I(X_1,X_2;\widehat{X}_1,\widehat{X}_2) & \numeq{1} I(X_1,X_2;\widehat{X}_1,\widehat{X}_2,\overline{X}_1^{\mathrm{cm}},\overline{X}_2^{\mathrm{cm}})\nonumber \\ 
&\numeq{2} I(X_1,X_2;\widehat{X}_1,\widehat{X}_2|\overline{X}_1^{\mathrm{cm}},\overline{X}_2^{\mathrm{cm}}) \nonumber \\&\quad + I(X_1,X_2;\overline{X}_1^{\mathrm{cm}},\overline{X}_2^{\mathrm{cm}})\nonumber \\
& \geab{3} I(X_1,X_2;\overline{X}_1^{\mathrm{cm}},\overline{X}_2^{\mathrm{cm}}), \label{eq:LB}
\end{align}
where \((1)\) is due to $\overline{X}_i^{\mathrm{cm}},i=1,2$, are  functions of $(\widehat{X}_1,\widehat{X}_2)$,  \((2)\) is due to the chain rule of mutual information, and \((3)\)   is due to $I(X_1,X_2;\widehat{X}_1,\widehat{X}_2|\overline{X}_1^{\mathrm{cm}},\overline{X}_2^{\mathrm{cm}})\geq 0$.  Thus,  (\ref{eq:MutualInfoIneq}) is obtained. If  $g_i^{\mathrm{cm}} (\widehat{X}_1,\widehat{X}_2)  = \widehat{X}_i-\text{a.s},\quad i=1,2$,  hold,   then $I(X_1,X_2;\widehat{X}_1,\widehat{X}_2|\overline{X}_1^{\mathrm{cm}},\overline{X}_2^{\mathrm{cm}})=0$, and hence the inequality (\ref{eq:LB}) become equality.  
(\textit{b}) The inequality is well-known, due to the orthogonal projection theorem. 
(c) This is due to (a), (b).
\end{proof}

\section{Structural Properties of Test Channels and Characterization of Joint RDF for Multivariate Jointly Gaussian Sources}
This section makes use of  Theorem~\ref{thm:jointlower} to derive additional structural properties of test channels for the joint RDF  $R_{X_1,X_2}(\Delta_1,\Delta_2)$ of  jointly Gaussian sources with square-error distortions, defined by (\ref{prob_1})-(\ref{prob_9}). 


\begin{theorem}[Sufficient conditions for the lower bounds of Theorem \ref{thm:jointlower} to be achieved] \label{thm:suffjoint} 
Consider the quadruple of zero mean   RVs $(X_1, X_2, \widehat{X}_1,\widehat{X}_2)$ taking values in  ${\mathbb R}^{p_1} \times {\mathbb R}^{p_2}\times {\mathbb R}^{p_1}\times {\mathbb R}^{p_2}$, $p_1,p_2 \in {\mathbb Z}_+$, with jointly Gaussian  distribution i.e,  ${\bf P}_{X_1,X_2, \widehat{X}_1,\widehat{X}_2}$ = ${\bf P}^G_{X_1,X_2, \widehat{X}_1,\widehat{X}_2}$ and ${{\mathbb X}}_1 \times {{\mathbb X}}_2-$joint marginal the fixed distribution ${\bf P}_{X_1,X_2}$ of $(X_1, X_2)$. Define the vectors,
\begin{align}
X = \begin{pmatrix}
X_1 \\X_2
\end{pmatrix},\;\widehat{X} = \begin{pmatrix}
\widehat{X}_1\\\widehat{X}_2
\end{pmatrix},\;\overline{X}^{\mathrm{cm}} \tri   \mathbf{E}\Bigg\{ \begin{pmatrix}
X_1 \\ X_2 
\end{pmatrix} \Big| \widehat{X}  \Bigg\}= \begin{pmatrix} \overline{X}_1^{\mathrm{cm}} \\ \overline{X}_2^{\mathrm{cm}}\end{pmatrix} .\nonumber
\end{align}
(a) If  the vector of conditional means satisfy, 
\begin{align*}
\overline{X}^{\mathrm{cm}} & = \mathbf{E}\big \{X\big \} + \Cov\big(X,\widehat{X}\big)\big\{\Cov \big(\widehat{X},\widehat{X}\big)\big\}^{\dagger}\Big (\widehat{X} - \mathbf{E} \big\{\widehat{X}\big \} \Big )  = \widehat{X}      
\end{align*}
 where $\dagger$ denotes pseudoinvesrse,   then  the equalities hold:
\begin{align}
&\overline{X}_1^{\mathrm{cm}} \tri  \mathbf{E}\Big\{X_1\Big|\widehat{X} \Big\} = \widehat{X}
_1,\;\; \overline{X}_2^{\mathrm{cm}} \tri  \mathbf{E}\Big\{X_2\Big|\widehat{X}\Big\} = \widehat{X}
_2. \label{eq:MSESuff2}
\end{align}
(b) If  the inverse of $\Cov \big(\widehat{X},\widehat{X}\big)$ exists and  $\mathbf{E}\big \{X\big \}=\mathbf{E} \big\{\widehat{X}\big \}=0$, then (\ref{eq:MSESuff2}) holds if  Condition 1  holds:
\begin{align} 
&\text{{Condition 1.}} \;\;\;\Cov\big(X,\widehat{X}\big) \big\{\Cov \big(\widehat{X},\widehat{X}\big)\big\}^{-1} = I_{p_1+p_2}. \label{eq:MSEcon}
\end{align}
(c) The lower bounds of Theorem~\ref{thm:jointlower} are achieved, if  there exist $(\widehat{X}_1, \widehat{X}_2)$ such that $\overline{X}^{\mathrm{cm}}=\widehat{X}$, or the statement of (b)   holds. 

\end{theorem}
\begin{proof} Follows by properties of jointly  Gaussian RVs.
\end{proof}
In the next lemma,  we apply Theorem~\ref{thm:jointlower} and Theorem~\ref{thm:suffjoint} to find  a parametric jointly Gaussian realization of $(\widehat{X}_1,\widehat{X}_2)$, that induces  the set of  test channels of the joint RDF  $R_{X_1,X_2}(\Delta_1,\Delta_2)$ for  (\ref{prob_1})-(\ref{prob_9}). 

\begin{lemma}[Preliminary parametrization of test channel]
\label{lemma_real}
\noindent Consider the joint RDF  $R_{X_1,X_2}(\Delta_1,\Delta_2)$  for  (\ref{prob_1})-(\ref{prob_9}). The following hold. \\
(a) A  jointly Gaussian distribution  $\mathbf{P}_{X_1,X_2,\widehat{X}_1,\widehat{X}_2}$  minimizes $I(X_1,X_2;\widehat{X}_1,\widehat{X}_2)$, subject to two  average distortions.\\
(b) The test channel distribution $\mathbf{P}_{\widehat{X}_1,\widehat{X}_2|X_1,X_2}$ of the joint RDF $R_{X_1,X_2}(\Delta_1,\Delta_2)$  is induced by the parametric Gaussian realization of $(\widehat{X}_1,\widehat{X}_2)$,  in  terms of the matrices $(H,Q_V)$, as
\begin{align}
&\widehat{X} = HX + V \label{eq:RealGaussModel1}\\
&H \in \mathbb{R}^{(p_1+p_2)\times (p_1+p_2)}, \quad V: \Omega \rightarrow \mathbb{R}^{(p_1+p_2)}, \\
& V \in G(0,Q_{(V_1,V_2)}),\; Q_{(V_1,V_2)} \succeq 0, \;
\mbox{$V$ and $X$ indep.},\label{eq:RealGaussModel2}
\end{align}
(c)  Consider part (b) and suppose there exist matrices $(H,Q_{(V_1,V_2)})$ such that Theorem \ref{thm:suffjoint}.(a) holds, i.e., $\overline{X}^{\mathrm{cm}} = \widehat{X}$-a.s., or in the special case Condition 1 holds.  Then the infimum in  $R_{X_1,X_2}(\Delta_1,\Delta_2)$  is taken over the subset $\pazocal{M}^{\mathrm{cm}, G}(\Delta_{1},\Delta_{2}) \subseteq \pazocal{M}^{\mathrm{cm}}(\Delta_{1},\Delta_{2})$,   
\begin{align}
\pazocal{M}^{\mathrm{cm}, G}&(\Delta_{1},\Delta_{2}) \tri  \Big \{\widehat{X}: \Omega \rightarrow \mathbb{R}^{(p_1+p_2)} \Big |\; \eqref{eq:RealGaussModel1}-\eqref{eq:RealGaussModel2} \; \mbox{hold},   \nonumber \\
&\;\;\overline{X}_i^{\mathrm{cm}} = \widehat{X}_i,\;  \mathbf{E} \big\{ ||X_i-\widehat{X}_i||_{\mathbb{R}^{p_i}}^2\big \} \le \Delta_{i}, i=1,2\Big \}
\end{align}
\end{lemma}
\begin{proof}(a) This  is similar to the classical RDF $R_{X}(\Delta)$ of a Gaussian RV $X\in G(0,Q_{X})$ with square-error distortion. (b) By part (a), the test channel distribution $\mathbf{P}_{\widehat{X}_1,\widehat{X}_2|X_1,X_2} $ is conditionally Gaussian with linear conditional mean $\mathbf{E} \big \{ X|\widehat{X}\big \}$ and non-random covariance $\Cov(X,\widehat{X}|X)$. Such a distribution is induced by the realizations \eqref{eq:RealGaussModel1}-\eqref{eq:RealGaussModel2}. (c) Follows from Theorem \ref{thm:suffjoint}.(c).
\end{proof}
Next, we construct   $(H,Q_{(V_1,V_2)})$ such that   $\overline{X}_i^{\mathrm{cm}} = \mathbf{E} \big \{ X_i |\widehat{X} \big \} = \widehat{X}_i-a.s$ for $i=1,2$, and  characterize  $R_{X_1,X_2}(\Delta_1,\Delta_2)$.

\begin{theorem}[Realization of optimal test channels and characterization of joint RDF] \label{thm:JointRDFChara}
Consider the joint RDF  $R_{X_1,X_2}(\Delta_1,\Delta_2)$  for  (\ref{prob_1})-(\ref{prob_9}).  \\
(a) The  test channel distribution $\mathbf{P}_{\widehat{X}_1,\widehat{X}_2|X_1,X_2}$ of the RDF $R_{X_1,X_2}(\Delta_1,\Delta_2)$ is induced by the parametric realization \eqref{eq:RealGaussModel1}-\eqref{eq:RealGaussModel2}, where the matrices, $(H,Q_V)$ satisfy, 
\begin{align}
 HQ_{(X_1,X_2)} &= Q_{(X_1,X_2)} - \Sigma_{(E_1,E_2)} =Q_{(X_1,X_2)}H\T \succeq 0,  \label{reali_1} \\
 Q_{(V_1,V_2)} &= HQ_{(X_1,X_2)}-H Q_{(X_1,X_2)} H\T \succeq 0.
  \label{reali_1_a}
\end{align}
Moreover,  $R_{X_1,X_2}(\Delta_1,\Delta_2)$ is characterized  by,
\begin{align}
&R_{X_1,X_2}(\Delta_1,\Delta_2)
= \inf_{{\pazocal{Q}}^\dagger(\Delta_1,\Delta_2)}    \frac{1}{2}\log \Big\{\frac{\det\big( Q_{(X_1,X_2)}\big)}{\det\big(\Sigma_{(E_1,E_2)}\big) } \Big\}  , \label{eq:JointRDFOpti}\\
&{\pazocal{Q}}^\dagger(\Delta_1,\Delta_2) \tri  \Big \{\Sigma_{(E_1,E_2)}: \:  (H, Q_{(V_1,V_2)}) \: \: \mbox{satisfy (\ref{reali_1}), (\ref{reali_1_a})},   \nonumber \\
&   \hspace{2.3cm} \trace \big ( \Sigma_{E_1}\big) \le \Delta_1,  \hso   \trace \big ( \Sigma_{E_2}\big) \le \Delta_2\Big \}. 
\end{align}
(b) Suppose  $Q_{(X_1,X_2)}\succ 0$. If $R_{X_1,X_2}(\Delta_1,\Delta_2)<\infty$,  then  the matrices, $(H,Q_{(V_1, V_2)})$, of part (a) reduce to, 
\begin{align}
&H = I_{p_1+p_2} - \Sigma_{(E_1,E_2)}Q_{(X_1,X_2)}^{-1},   \label{eq:Joint_inv_1}  \\
& Q_{(V_1,V_2)} = \Sigma_{(E_1,E_2)} -\Sigma_{(E_1,E_2)}Q_{(X_1,X_2)}^{-1}\Sigma_{(E_1,E_2)} \succeq 0, \label{eq:Joint_inv_2}\\
&Q_{(X_1,X_2)}- \Sigma_{(E_1,E_2)}\succeq 0,\hst \Longleftrightarrow  \label{eq:Joint_inv_3}      \\
&\Sigma_{(E_1,E_2)} -\Sigma_{(E_1,E_2)}Q_{(X_1,X_2)}^{-1}\Sigma_{(E_1,E_2)} \succeq 0.\label{eq:Joint_inv_4}     
\end{align}
and  ${\pazocal{Q}}^\dagger(\Delta_1,\Delta_2)$ in   (\ref{eq:JointRDFOpti}) is   replaced by  $\sr{\circ}{\pazocal{Q}}(\Delta_1,\Delta_2)$,   given by
\begin{align}
\sr{\circ}{\pazocal{Q}}(\Delta_1,\Delta_2) \tri & \Big \{\Sigma_{(E_1,E_2)}:  \: \: Q_{(X_1,X_2)} \succeq \Sigma_{(E_1,E_2)} \succeq 0,  \nonumber \\ 
&  \;\trace \big ( \Sigma_{E_1}\big) \le \Delta_1, \trace \big ( \Sigma_{E_2}\big) \le \Delta_2 \Big \}. \label{eq:jointConstr}
\end{align}
\end{theorem}
\begin{proof}
See Appendix \ref{App:realizations}.
\end{proof}
\begin{lemma}
\label{lem:char}
Consider  $R_{X_1,X_2}(\Delta_1,\Delta_2)$ of Theorem~\ref{thm:JointRDFChara}, defined by  (\ref{eq:JointRDFOpti}) and assume   $Q_{(X_1, X_2)}\succ 0$, and  $R_{X_1,X_2}(\Delta_1,\Delta_2)< +\infty$. The Lagrange functional is, 
\begin{align}
&{\cal L}\tri  \frac{1}{2}\log\Big\{ \frac{\det\big( Q_{(X_1,X_2)}\big)}{\det\big(\Sigma_{(E_1,E_2)}\big) } \Big\} + \trace \Big (\Theta \Big( \Sigma_{(E_1,E_2)} - Q_{(X_1,X_2)} \Big)\Big )    \nonumber \\
&+ \lambda_1 \Big ( \trace \Big (\Sigma_{E_1}\Big )  - \Delta_{1} \Big ) + \lambda_2 \Big ( \trace \Big (\Sigma_{E_2}\Big )  - \Delta_{2} \Big )  -\trace \Big (V\Sigma_{(E_1,E_2)} \Big ) \nonumber
\end{align}
where $ \Theta \succeq 0$, $V \succeq 0$, $\lambda_i \in [0,\infty), i=1,2$. The optimal $\Sigma_{(E_1, E_2)}  \in \sr{\circ}{\pazocal{Q}}(\Delta_{1},\Delta_{2})$ for $ R_{X_1,X_2}(\Delta_{1},\Delta_{2})$ is found as follows.  \\
(i) Stationarity:
\begin{align}
-\frac{1}{2} \Sigma_{(E_1,E_2)}^{-1}  +  \begin{bmatrix}
\lambda_1 I_{p_1} & 0 \\ 0 & \lambda_2 I_{p_2}
\end{bmatrix} + \Theta +V=0 . \label{eq:ErrorCovLag}
\end{align}
(ii) Complementary Slackness:
\begin{align}
&\lambda_1 \Big ( \trace \Big (\Sigma_{E_1}\Big )  - \Delta_{1} \Big ) = 0,\;\; \lambda_2 \Big ( \trace \Big (\Sigma_{E_2}\Big )  - \Delta_{2} \Big ) = 0, \label{eq:CompSlackLambda} \\
&\trace \Big (V\Sigma_{(E_1,E_2)}\Big ) = 0,\; \trace \Big (\Theta \Big( \Sigma_{(E_1,E_2)} - Q_{(X_1,X_2)} \Big)\Big ) = 0. \label{eq:CompSlackTheta}
\end{align}
(iii) Primal Feasibility: Defined  by $\sr{\circ}{\pazocal{Q}}(\Delta_{1},\Delta_{2})$.\\
(iv) Dual Feasibility: $\lambda_1 \geq 0, \hso \lambda_2 \geq 0, \hso \Theta \succeq 0, \hso V \succeq 0$.\\ Moreover, the following hold.\\
(a)  $V=0$, and   
\begin{align}
 \Sigma_{(E_1,E_2)}  = \frac{1}{2}\Bigg ( \begin{bmatrix}
\lambda_1 I_{p_1} & 0 \\ 0 & \lambda_2 I_{p_2}
\end{bmatrix} + \Theta \Bigg)^{-1} \succ 0 . \label{eq:ErrorCovLag_new}
\end{align}
(b) If  $Q_{(X_1,X_2)}-\Sigma_{(E_1,E_2)}  \succ 0$ 
then $\Theta=0$,  and 
\begin{align}
 \Sigma_{(E_1,E_2)}  = \frac{1}{2}\Bigg ( \begin{bmatrix}
\lambda_1 I_{p_1} & 0 \\ 0 & \lambda_2 I_{p_2}
\end{bmatrix}  \Bigg)^{-1} \succ 0 . \label{eq:ErrorCovLag_new_1}
\end{align}
\end{lemma} 
\begin{proof} The derivation is standard hence it is omitted.
\end{proof}

The next two theorems are obtained from  Lemma~\ref{lem:char}.

\begin{theorem}[Joint RDF for a positive surface]
 \label{thm:JointRDFGen}
Consider the characterization of  joint RDF $R_{X_1,X_2}(\Delta_1,\Delta_2)$ of Theorem~\ref{thm:JointRDFChara}, defined by  (\ref{eq:JointRDFOpti}), and  assume  $Q_{(X_1, X_2)}\succ 0$ (i.e., this implies $Q_{X_1}\succ 0, Q_{X_2}\succ 0$).
Define the set 
\begin{align}
{\cal D}_{(X_1,X_2)} = \bigg \{  &(\Delta_{1},\Delta_{2}) \in [0,\infty) \times[0,\infty) \bigg |   Q_{(X_1,X_2)}-\Sigma_{(E_1,E_2)} \succ 0 \bigg \}.  \nonumber
\end{align}
The joint RDF $R_{X_1,X_2}(\Delta_1,\Delta_2)$ for $( \Delta_1,\Delta_2\big ) \in\mathcal{D}_{(X_1,X_2)}$ is 
 \begin{align*}
&R_{X_1,X_2}\big (\Delta_1,\Delta_2 \big ) =\frac{1}{2}\log \bigg\{\frac{\det\big( Q_{(X_1,X_2)}\big)}{\det\big(\Sigma_{E_1}\big)\det\big(\Sigma_{E_2}\big) } \bigg\}=\mbox{(\ref{equality})}     \\
&\Sigma_{E_1} = \diag \Big (\frac{\Delta_1}{p_1} ,\ldots, \frac{\Delta_1}{p_1} \Big ), \;\; \Sigma_{E_2} = \diag \Big (\frac{\Delta_2}{p_2} ,\ldots, \frac{\Delta_2}{p_2} \Big )
      \end{align*}
and this  is achieved by the covariance matrix $ \Sigma_{(E_1,E_2)}$ with   $ \Sigma_{E_1,E_2}=Q_{X_1,X_2|\widehat{X}}=0$, and Gray's lower bound (\ref{equality}) holds.
\end{theorem}
\begin{proof} For any element of the set ${\cal D}_{(X_1,X_2)}$ then $Q_{(X_1,X_2)} - \Sigma_{(E_1,E_2)} \succ 0$, and the statements follow from   Lemma~\ref{lem:char}. 
\end{proof}

\begin{remark}
 For the scalar-valued RVs, i.e.,    $p_1=p_2=1$, we have verified that Lemma~\ref{lem:char} produces the  closed-form expression of $R_{X_1,X_2}(\Delta_1,\Delta_2)$ as derived in  \cite[Theorem 6]{xiao}. However, for the multivariate case of Lemma~\ref{lem:char}, to obtain the  closed-form expression is challenging. To make the problem tractable, in  Theorem~\ref{the:cvf},   we use the  canonical variable form of the tuple $(X_1, X_2)$, as described in  \cite{charalambous2020characterization} and \cite{anewapproach}.
\end{remark}
\noindent The algorithm to transform the tuple $(X_1,X_2)$ to the canonical variable form is presented below \cite[Algorithm 2.10]{anewapproach}. 
\begin{Algorithm}{ \cite[Algorithm 2.10]{anewapproach} Transformation of a variance matrix to its canonical variable form. \label{alg:cvf}\\}
Data : $p_1,p_2 \in Z_+,\;Q_{(X_1,X_2)} \in \mathbb{R}^{(p_1+p_2)\times(p_1+p_2)}$, satisfying $Q_{(X_1,X_2)} = Q_{(X_1,X_2)}^T \succ  0$, with decomposition \eqref{prob_1}.\\
1) Perform singular value decompositions (SVD), $Q_{X_i} = U_iD_iU_i^T,\;i=1,2$ with $U_i\in \mathbb{R}^{p_i\times p_1},\;$ orthogonal and $D_i = \diag(d_{1,1},\dots,d_{i,p_{i}})\in \mathbb{R}^{p_1\times p_1},\;\;d_{i,1} \ge d_{i,2} \ge \dots \ge d_{i,p_i}>0$. \\
2) Perform SVD of $
D_1^{-\frac{1}{2}}U_1^T Q_{X_1,X_2}U_2D_2^{-\frac{1}{2}} = U_3D_3U_4^T$ with $U_3\in \mathbb{R}^{p_1\times p_1},\;U_4\in \mathbb{R}^{p_2\times p_2}$ orthogonal and  \renewcommand\arraystretch{0.92}
\begin{align}
&D_3 = \mathrm{Block}\text{-}\diag\big(I_{p_{11}},D_4,0_{p_{13} \times  p_{23}} \big)\in \mathbb{R}^{p_1\times p_2},\nonumber \\&D_4 = \diag(d_{4,1},\dots,d_{4,p_{12}})\in \mathbb{R}^{p_{12}\times p_{22}},\;1 > d_{4,1} \ge \dots \ge d_{4,p_{12}}>0, \nonumber \\
& p_i = p_{i1} + p_{i2} + p_{i3},\;\;i=1,2,\quad p_{11} = p_{21},\quad p_{12} = p_{22} \nonumber
\end{align}
3) Compute the new variance matrix and the transformation to the canonical variable representation $(X_1 \mapsto S_1X_1,\;X_2\mapsto S_2X_2)$ according to  \renewcommand\arraystretch{1.3}
\begin{align}
Q_{\mathrm{\mathrm{cvf}}} = \begin{pmatrix}
I_{p_1} & D_3 \\
D_3^T & I_{p_2}
\end{pmatrix},\quad S_1 =U_3^TD_1^{-\frac{1}{2}}U_1^T,\;\;S_2 = U_4^TD_2^{-\frac{1}{2}}U_2^T \nonumber
\end{align}
\end{Algorithm}

 \begin{theorem} 
 \label{the:cvf}
Consider the statement of Theorem~\ref{thm:JointRDFChara}.(b), with $(X_1 , X_2 ) \in G(0,Q_{(X_1,X_2)}), Q_{(X_1,X_2)}\succ 0$. Determine the canonical variable form of the tuple $(X_1 , X_2 ) $, according to \cite[Definition~2.2]{anewapproach} by using algorithm Algorithm \ref{alg:cvf}, and restrict attention to indices, $p_{11} = p_{21}=0$. Then,
 $n = p_{12} =p_{22}$ and $p_1=p_{12}+p_{13},  p_2=p_{22}+p_{23}$. Similarly, transform $(E_1 , E_2 ) \in G(0,\Sigma_{(E_1,E_2)})$ with  $\overline{p}_{11} = \overline{p}_{21}=0$ and 
$\overline{n} = \overline{p}_{12} =\overline{p}_{22}$, $\overline{p}_1=\overline{p}_{12}+\overline{p}_{13},  \overline{p}_2=\overline{p}_{22}+\overline{p}_{23}$. \\ 
The joint RDF $R_{X_1,X_2}(\Delta_1, \Delta_2)$ of  Theorem~\ref{thm:JointRDFChara}.(b),
is equivalently characterized by  \vspace{-0.2cm}
\begin{align*}
&R_{X_1,X_2}(\Delta_1, \Delta_2)= \inf_{\sr{\circ}{\pazocal{Q}}(\Delta_1,\Delta_2)}  \frac{1}{2} \log \Big\{ \frac{\det\big(D_1
\big) \det \big(D_2\big)\det \big(Q_{\mathrm{cvf}}\big)}{\det \big(\overline{D}_1\big) \det \big(\overline{D}_2\big) \det \big( \Sigma_{\mathrm{cvf}} \big) }\Big\}
\end{align*}
where, \vspace{-0.1cm}
\begin{align*}
&\sr{\circ}{\pazocal{Q}}(\Delta_1,\Delta_2)\tri \Big\{\overline{n} \in {\mathbb Z}_+, \; \overline{d}_{4,i} \in (0,1),\; i=1, \ldots, \overline{n}, \nonumber \\
& \hspace{1.18cm}\overline{d}_{1,i}\in (0,\infty),\; i=1, \ldots, \overline{p}_1, \;\overline{d}_{2,i}\in (0,\infty), \;i=1, \ldots, \overline{p}_2: \nonumber \\
& \hspace{1.1cm} \sum_{i=1}^{\overline{p}_1} \overline{d}_{1,i}   \leq \Delta_1, \;  \sum_{i=1}^{{\overline{p}}_2} \overline{d}_{2,i}   \leq \Delta_2, \; Q_{(X_1,X_2)}- \Sigma_{(E_1,E_2)}\succeq  0 \Big\}\\
\end{align*} 
and
\begin{align*}
    &     \det\big(\Sigma_{\mathrm{\mathrm{cvf}}}\big)
   =  \det\big(I_{\overline{p}_1}-\overline{D}_3 \overline{D}_3\T\big) \\
 &  = \left\{
        \begin{array}{lll}
          1, & \mbox{if} & \overline{p}_{13} > 0, ~ \overline{p}_{23} > 0, ~
                 \overline{p}_{12} = \overline{p}_{22} = 0, \\
          \prod_{i=1}^{\overline{n}} \left(
              1-\overline{d}_{4,i}^2            \right), & \mbox{if} & \overline{p}_{12} = \overline{p}_{22} =\overline{n}, ~
                        \overline{p}_{13} \geq 0, ~ \overline{p}_{23} \geq 0.
        \end{array}
        \right.
\end{align*}    
\end{theorem}
\begin{proof} By Theorem~\ref{thm:JointRDFChara}.(b) and applying  \cite[Definition~2.2]{anewapproach} and  Algorithm \ref{alg:cvf} we obtain the results. 
\end{proof}

\begin{figure}[t]
  \centering
  \includegraphics[width = 0.99\columnwidth]{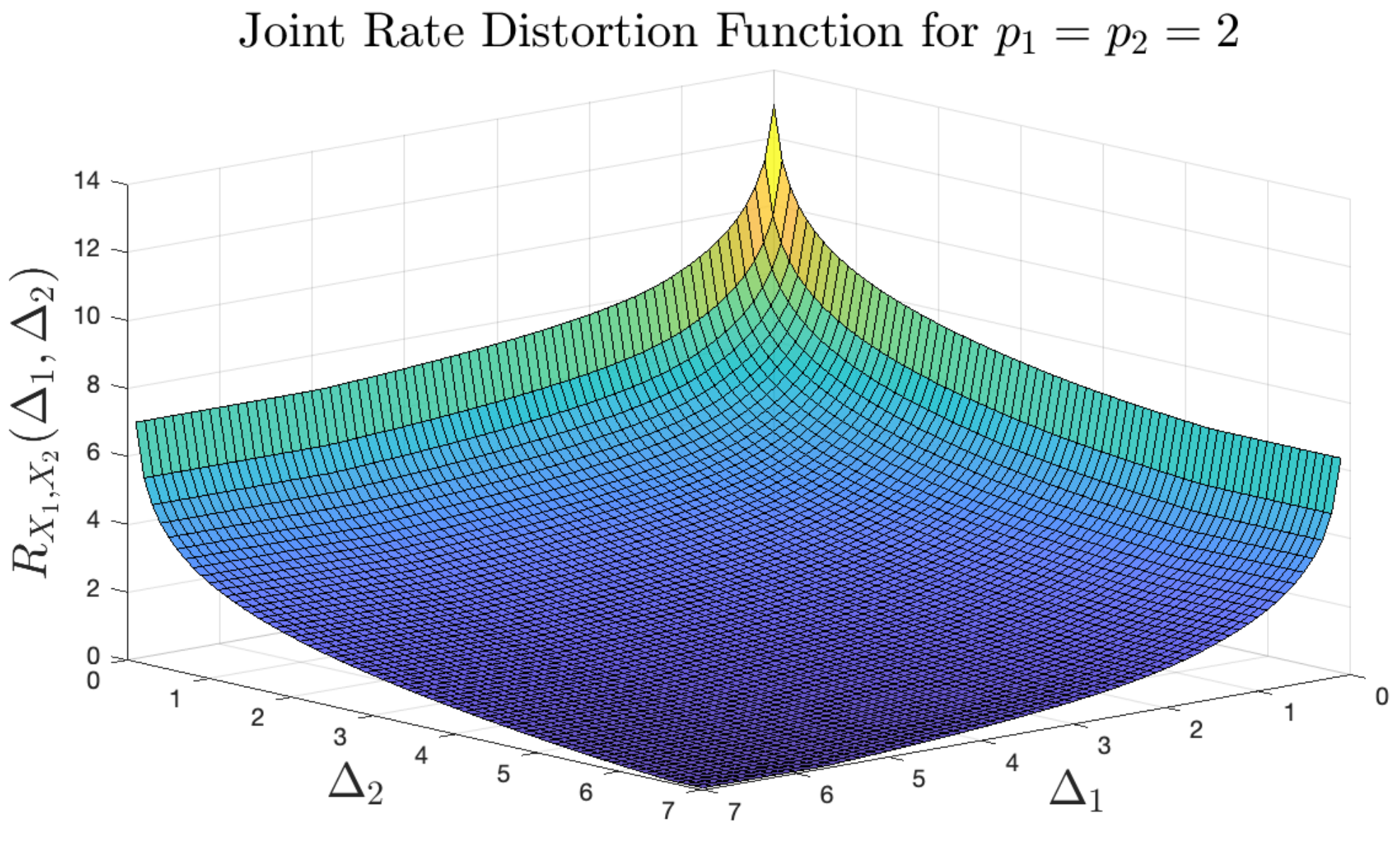}
  \vspace{-0.5cm}
\caption{Joint RDF $R_{X_1,X_2}(\Delta_1,\Delta_2)$ of source of  Section~\ref{sect:num},  $p_1=p_2=2$\vspace{-0.5cm}\vspace{-0.55cm}.}
  \label{fig:jointRDFplot}
\end{figure}

\begin{remark} $R_{X_1,X_2}(\Delta_1, \Delta_2)$ of Theorem~\ref{the:cvf}, is much easier to optimize, due to its structure.
\end{remark}

\section{Evaluation of the Joint RDF via SDP}
\label{sect:num}
We can express the optimization problem of Theorem~\ref{thm:JointRDFChara}  as a semidefinite program (SDP) as follows, define $\Xi_1\T =  \mathrm{Block}\text{-}\mathrm{diag} \big ( I_{p_1}\;0_{p_2} \big )$ and $\Xi_2\T =  \mathrm{Block}\text{-}\mathrm{diag} \big (0_{p_1} \; I_{p_2} \big )$,
    \begin{mini}[2]
	  {\Sigma_{(E_1,E_2)}}{ \frac{1}{2}\log\Big\{ \frac{\det\big( Q_{(X_1,X_2)}\big)}{\det\big(\Sigma_{(E_1,E_2)}\big) } \Big\} }{}{} 
	 \hspace{1cm} 	  	  	 \addConstraint{Q_{(X_1,X_2)}  - \Sigma_{(E_1,E_2)}\succeq 0,\;\; \Sigma_{(E_1,E_2)}}{\succeq 0} \label{eq:SDP}
	  \addConstraint{\trace\big (\Xi_i\T\Sigma_{(E_1,E_2)}\Xi_i\big )}{\le \Delta_i,\;i=1,2}
    \end{mini}
Then, we can solve the SDP \eqref{eq:SDP} by using the CVX~\cite{cvx}. Below,  we  calculate the optimal $\Sigma_{(E_1,E_2)}$ for a  multivariate example  $X_i: \Omega \rightarrow \mathbb{R}^2,\;i=1,2$,  with covariance,  \renewcommand\arraystretch{0.95}
\begin{align*}
{\small
Q_{(X_1,X_2)} = \begin{pmatrix}
3.929&-0.11&\vline&0.642&0.976\\
-0.11&2.629&\vline&-0.859&0.337\\
\hline
0.642&-0.859&\vline&2.142&1.797\\
0.976&0.337&\vline&1.797&3.495\\
\end{pmatrix} }.
\end{align*} 
Fig.~\ref{fig:jointRDFplot} depicts  $R_{X_1,X_2}(\Delta_1,\Delta_2), (\Delta_1,\Delta_2)\in [0,\infty)\times [0,\infty)$. 
Below we   distinguish two  cases.\\
{\it Case 1. } Given distortions $(\Delta_1,\Delta_2) = (0.4, 0.5)$, the solution of \eqref{eq:JointRDFOpti},  \eqref{eq:jointConstr} is given by  
\begin{align*}
\Sigma_{(E_1,E_2)} = \diag\big (0.2,0.2,0.25,0.25 \big ), \; Q_{(X_1,X_2)} - \Sigma_{(E_1,E_2)}\succ 0
\end{align*}
Distortions $\Delta_1$ and $\Delta_2$ are equally divided among the diagonal elements of the first and second 2-by-2 diagonal blocks of $\Sigma_{(E_1,E_2)}$ respectively, and the rest of the values are  zero. Hence,  $(0.4, 0.5)\in  \mathcal{D}_{(X_1,X_2)}$;  this  re-confirms  Theorem \ref{thm:JointRDFGen}.\\ {\it Case 2.} Given distortions $(\Delta_1,\Delta_2) = (1.65,1.85)$, the optimal error covariance matrix is given by, 
 \begin{align}
 {\small
\Sigma_{(E_1,E_2)} =\begin{pmatrix}
0.849&-0.0017&\vline&-0.0053&0.0036\\
-0.0017&0.801&\vline&-0.144&0.0961\\
\hline
-0.0053&-0.144&\vline&0.804&0.293\\
0.0036&0.0961&\vline&0.293&1.05\\
\end{pmatrix} } \nonumber
\end{align}
and  $Q_{(X_1,X_2)} - \Sigma_{(E_1,E_2)}\succeq 0$ but not positive definite. Unlike Case 1, $\Sigma_{(E_1,E_2)}$ is not block-diagonal, i.e., $\Sigma_{E_1,E_2} \neq 0$, as in Theorem \ref{thm:JointRDFGen}, hence $(1.65,1.85)\notin  \mathcal{D}_{(X_1,X_2)}$. This choice of distortions  corresponds to Lemma \ref{lem:char}.(b). 
\section{Conclusion}
The joint RDF ${R}_{X_1,X_2}(\Delta_1,\Delta_2)$, with individual  distortion criteria, is analyzed, with emphasis on the     structural properties of realizations of the  reproduction  RVs $(\widehat{X}_1,\widehat{X}_2)$ of $(X_1, X_2)$,  and   corresponding  optimal test channel distribution, $\mathbf{P}_{\widehat{X}_1,\widehat{X}_2|X_1,X_2}$.   Closed-form expressions of $R_{X_1,X_2}(\Delta_1,\Delta_2)$ are derived for a strictly positive surface of the distortion region,  and a numerical technique is presented, 
which verifies the closed-form expressions.


\section{Appendices}
{\it Proof of Theorem \ref{thm:JointRDFChara}}. \label{App:realizations}
Consider  \eqref{eq:RealGaussModel1}-\eqref{eq:RealGaussModel2}. To identify  $(H, Q_{(V_1,V_2)})$ such that  $ \overline{X}_i^{\mathrm{cm}} = {\bf E}\Big\{X_i\Big| \widehat{X}\Big\}=\widehat{X}_i, i=1,2$, we make use of the following preliminary calculations. The covariance of $X$ and  $\widehat{X}$ is,  
\begin{align}
Q_{X,\widehat{X}} & = \mathbf{E}\Big \{X\Big(HX + V\Big)\T\Big \} = Q_{(X_1,X_2)}H\T . \label{cross_1}
\end{align}
By \eqref{eq:RealGaussModel1}-\eqref{eq:RealGaussModel2},   the covariance of  $\widehat{X} = HX + V$ is
\begin{align}
Q_{\widehat{X}} &= \mathbf{E}\Big\{\widehat{X}\widehat{X}\T\Big\}  
=  HQ_{(X_1,X_2)}H\T + Q_{(V_1,V_2)},  \label{repr_1}
\end{align}
Consider the special case when \emph{Condition 1, \eqref{eq:MSEcon}} holds:
\begin{align}
&\Cov\big(X,\widehat{X}\big) \big\{\Cov \big(\widehat{X},\widehat{X}\big)\big\}^{-1} = I_{p_1 + p_2} \; \Longleftrightarrow \; Q_{X,\widehat{X}}Q_{\widehat{X}}^{-1} = I_{p_1 + p_2} \nonumber \\
&\Longrightarrow Q_{X}H\T = HQ_{X}H\T + Q_{(V_1,V_2)}  \hst \mbox{by (\ref{cross_1}),  (\ref{repr_1})}  \nonumber  \\
 &\Longrightarrow Q_{(V_1,V_2)} =Q_{(X_1,X_2)}H\T -HQ_{(X_1,X_2)}H\T . \label{eq:V}
\end{align}
Next, we turn to the identification of   $H$.  By the definition of covariance of the errors, then $\Sigma_{(E_1,E_2)} \tri \Cov(X,X|\widehat{X})$, and
\begin{align}
\Sigma_{(E_1,E_2)} &= \Cov(X,X) - \Cov(X,\widehat{X})\big\{\Cov(\widehat{X},\widehat{X})\big\}^{-1}\Cov(X,\widehat{X})\T  \nonumber\\
&=Q_{(X_1,X_2)} - HQ_{(X_1,X_2)}, \;\;  \mbox{by \eqref{eq:MSEcon},(\ref{cross_1})}   \nonumber \\
 &\hspace{-0.8cm }\Longrightarrow HQ_{(X_1,X_2)} =Q_{(X_1,X_2)}-\Sigma_{(E_1,E_2)} = Q_{(X_1,X_2)}H\T  \label{eq:H_1}\\
 & \hspace{-0.8cm }\Longrightarrow H = I_{p_1+p_2} - \Sigma_{(E_1,E_2)}Q_{(X_1, X_2)}^{-1}, \; \mbox{if $Q_{(X_1, X_2)}\succ 0$}.\nonumber 
\end{align}
Using  (\ref{eq:H_1}) into  (\ref{eq:V}) then we have 
\begin{align}
Q_{(V_1,V_2)} =& Q_{(X_1,X_2)}H\T -HQ_{(X_1,X_2)}H\T = Q_{(V_1,V_2)}\T \label{eq:V_1}\\
=&Q_{(X_1,X_2)}-\Sigma_{(E_1,E_2)} -HQ_{(X_1,X_2)}H\T. \label{eq:V_2}
\end{align}
Hence,  $(H, Q_{(V_1,V_2)})$ are obtained. The general case is shown by using properties of pseudoinverse. The rest follow.
\section{Acknowledgments}
\par This work was supported in parts by  the European Regional Development Fund and
 the Republic of Cyprus through the Research Promotion Foundation Projects EXCELLENCE/1216/0365 
and  EXCELLENCE/1216/0296. 

\newpage

\label{Bibliography}
\bibliographystyle{IEEEtran}
\bibliography{bibliography}

\end{document}